\documentclass[11pt,a4paper]{article}
\usepackage{preferred_direction}

\title{Hyperbolic inference for treatment comparisons\\
relative to a preferred direction:\\
An application to astigmatism axes}
\author{Buddhananda Banerjee\thanks{Corresponding author:
\href{mailto:bbanerjee@maths.iitkgp.ac.in}{bbanerjee@maths.iitkgp.ac.in}.},
Surojit Biswas and Daitari Prusty\\[6pt]
\small Department of Mathematics, Indian Institute of Technology Kharagpur\\
\small West Bengal, India--721302}
\date{20 November 2026}

\hypersetup{pdfauthor={Buddhananda Banerjee, Surojit Biswas, Daitari Prusty},
 pdfsubject={Research manuscript}}
\begin{document}
\maketitle
\begin{abstract}
The paper compares two independent von Mises populations relative to a known preferred
direction. An angular Fisher information calibration represents each population
as a point in a Poincar\'{e} disk. A signed score at that point combines the population's
mean departure from the medically  preferred direction with its concentration. We test
equality of scores using separate likelihood fits and an empirical sandwich
standard error. At every fixed pair of  concentrations the centered
statistic is asymptotically standard normal, including when either population
is aligned with or opposed to the target. The resulting test and confidence
interval concern the Busemann score difference. Simulations show that rejection under
the score null can exceed the nominal level at the smallest sample size and
is closer to nominal at larger sizes. Proposed method has larger empirical power for alternative with score contrast, compared to the existing one. 
An application compares the transformed
postoperative astigmatism axes for the SNARE and irrigating VECTIS procedures.
The estimated score is smaller for SNARE, but the confidence interval for the
score difference includes zero.
\end{abstract}

\vspace{5pt}
\section{Introduction}
\label{sec:introduction}
\subsection{A preferred direction changes the comparison}
\label{sec:preferred_coding}
Circular observations record directions, phases or orientations. Values
just below $2\pi$ and just above zero are close on the circle although
their ordinary numerical difference is large. Circular models account
for this periodic structure \citep{fisher1993circular,mardia2000directional,pewsey2021recent}.
Comparing two populations also requires a criterion that reflects the
scientific question.
Our motivating example uses the postoperative astigmatism endpoint of
\citet{biswas2016comparison}. They cite the thesis of
\citet{bakshi2010thesis} as the source of the clinical study. In their
clinical motivation, oblique distortion was regarded as more disruptive for
objects with horizontal and vertical features. Axes near either of
these meridians were therefore preferred. Let $\theta$ be the original
axial angle. The angles $\theta$ and $\theta+\pi$ describe the same axis.
Doubling gives an ordinary circular angle; it still separates the two
preferred meridians. Doubling again gives
\begin{equation}
 \Theta=4\theta\pmod{2\pi}.
 \label{eq:fourfold_coding}
\end{equation}
Thus the preferred original angles $0^\circ$, $90^\circ$ and $180^\circ$
all map to $0^\circ$ on the transformed circle. This is why we set
$\mu_0=0$ in the application. The fourfold coding treats axes that differ
by $90^\circ$ as equivalent. Closeness to zero on this scale represents
proximity to a preferred meridian. The endpoint does not measure the
magnitude of astigmatism.
More generally, the target is fixed by the application before fitting the
populations. Mean directions on opposite sides of it can be equally far
away. Equality of absolute departures therefore differs from equality of
mean directions, the hypothesis addressed by classical circular tests
such as that of \citet{watson1956significance}.

\subsection{Population criterion}
For a von Mises population with mean direction $\mu$ and concentration
$\kappa$, its absolute mean departure from the preferred direction
$\mu_0$ is $|\Arg\exp\{i(\mu-\mu_0)\}|$. This is a measure of location.
An alternative criterion is the mean cosine loss of individual responses.
With $A(\kappa)=I_1(\kappa)/I_0(\kappa)$, where $I_\nu$ is the modified
Bessel function of the first kind and order $\nu$, this loss is
\begin{equation}
 E\{1-\cos(\Theta-\mu_0)\}
   =1-A(\kappa)\cos(\mu-\mu_0).
 \label{eq:intro_cosine_loss}
\end{equation}
The angular and cosine criteria agree on their ordering at a fixed
concentration but can answer different questions when concentrations
differ. \citet{biswas2016comparison} develop tests for
comparisons relative to a preferred direction; related tests for
circular responses are studied by \citet{biswas2023alternative}.

We use a geometric score that depends on both the mean direction and
the concentration. Fisher information quantifies local likelihood
distinguishability \citep{amari2016information}. We match the angular
information coefficient of the von Mises model to that of a disk
representation. Each population is represented independently of its
sample size. The target direction is represented on the boundary,
without assigning it a concentration. A Busemann function gives a
finite score at each population point, and its level curves show which
populations have equal scores \citep{ghadimi2021busemann}. Section~\ref{sec:methodology}
defines the score and explains what its equality means.
The test estimates the two concentrations separately. It combines
directional and concentration uncertainty through an empirical
sandwich standard error. Section~\ref{sec:inference} shows that the
contrast has positive first-order variance throughout the fixed
positive-concentration parameter space, including at alignment and
opposition. Section~\ref{sec:simulation} assesses calibration and power;
Section~\ref{sec:data_analysis} applies the method to the transformed
cataract-axis endpoint. The appendices give the geometric calculations,
inferential proofs, common-concentration specialization, and details of
the data and simulation study.

\section{Population representation and comparison criterion}
\label{sec:methodology}
\subsection{Sampling model and Fisher information}
\label{sec:model}
For $g=1,2$, let $\Theta_{g1},\ldots,\Theta_{gn_g}$ be independent
observations from a von Mises population with density
\begin{equation}
 f(\theta;\mu_g,\kappa_g)
 =\frac{\exp\{\kappa_g\cos(\theta-\mu_g)\}}{2\pi I_0(\kappa_g)},
 \quad 0\leq\theta<2\pi,\quad \kappa_g>0.
 \label{eq:model}
\end{equation}
The two groups are independent, and the parameters
$\eta_g=(\mu_g,\kappa_g)^\top$ are estimated separately. There is no
equality restriction on the two concentrations. The preferred direction
$\mu_0$ is fixed and known. We write
$\phi=\Arg\exp\{i(\mu-\mu_0)\}$ when a representative angular departure
is useful; formulas involving $\cos\phi$ and $\sin\phi$ are periodic and
do not depend on the chosen angular cut.

For $\kappa>0$, define
\begin{equation}
 \begin{aligned}
 A(\kappa)&=\frac{I_1(\kappa)}{I_0(\kappa)},\\
 B(\kappa)&=A'(\kappa)=1-A(\kappa)^2-\frac{A(\kappa)}{\kappa},\\
 C(\kappa)&=\kappa A(\kappa).
 \end{aligned}
 \label{eq:information_functions}
\end{equation}
The per-observation information and its associated local metric are
\begin{equation}
 \mathcal I(\mu,\kappa)=\diag\{C(\kappa),B(\kappa)\},
 \qquad ds_F^2=C(\kappa)\,d\mu^2+B(\kappa)\,d\kappa^2.
 \label{eq:fisher_information}
\end{equation}
Both coefficients are positive. In particular, $B(\kappa)$ is the
variance of $\cos(\Theta-\mu)$, and
$C'(\kappa)=A(\kappa)+\kappa B(\kappa)>0$.
Appendix~\ref{app:information} derives these identities.

\subsection{The disk representation}
\label{sec:embedding}
Represent a population by
\begin{equation}
 z_F(\mu,\kappa)=R_F(\kappa)e^{i\mu},\qquad
 R_F(\kappa)=
 \frac{\sqrt{C(\kappa)}}{1+\sqrt{1+C(\kappa)}}.
 \label{eq:embedding}
\end{equation}
The radius increases from zero to one as concentration increases from
zero to infinity. At $\kappa=0$, all mean labels describe the uniform
distribution and map to the origin. At positive concentration, changing
$\mu$ rotates the point and changing $\kappa$ moves it radially.
Figure~\ref{fig:embedding} links this parameter-space representation to
the densities on the observation circle.

\begin{figure}[!htbp]
 \centering
 \includegraphics[width=.93\linewidth]{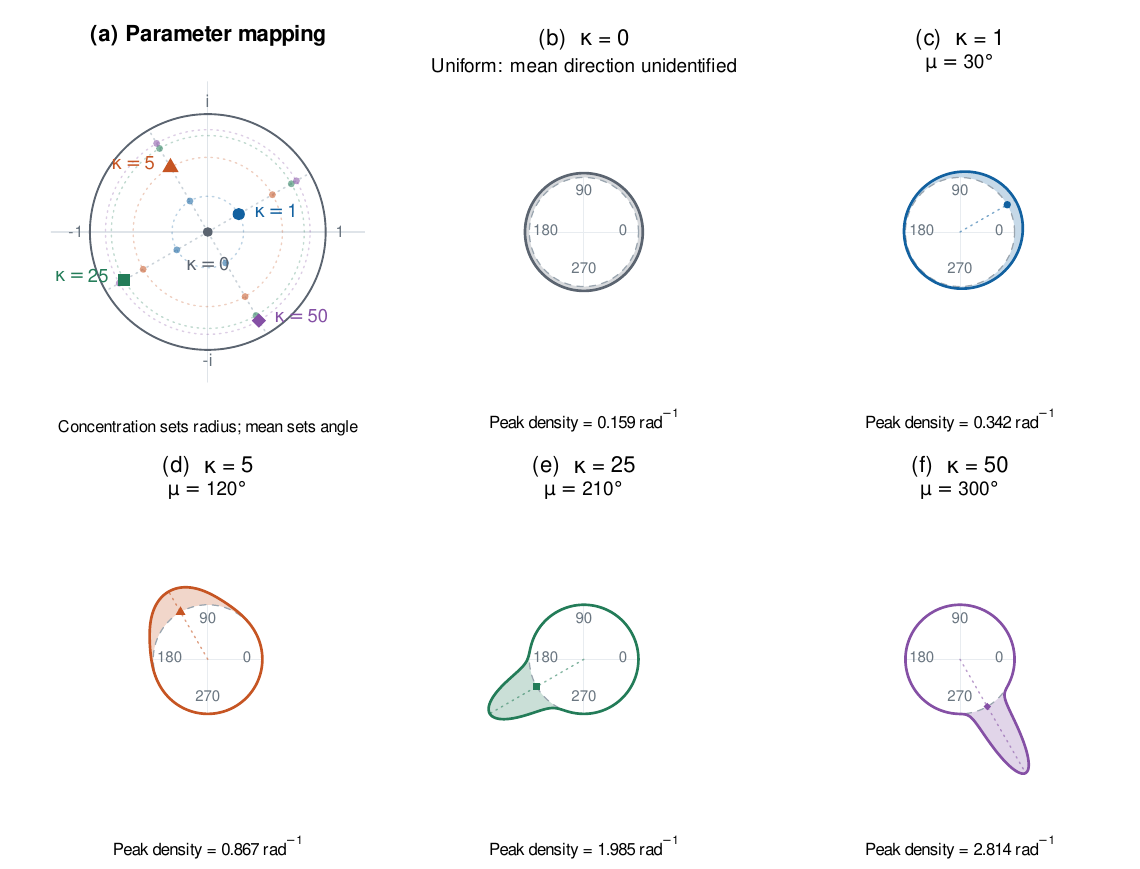}
 \caption{The von Mises representation. Panel (a) shows disk
 locations for several directions and concentrations. Panels (b)--(f)
 show circular densities for $\kappa=0,1,5,25,50$; the outward density
 height uses the same scale in all panels. Uniform populations map to
 the disk origin. The observation circles and the disk parameter space
 have different roles. All displayed populations are constructed examples.}
 \label{fig:embedding}
\end{figure}
Equip the disk $\mathbb D=\{z:|z|<1\}$ with the
Poincar\'e metric $ds_H^2=4|dz|^2/(1-|z|^2)^2$.
The radius in \eqref{eq:embedding} satisfies
\begin{equation}
 \frac{4R_F(\kappa)^2}{\{1-R_F(\kappa)^2\}^2}=C(\kappa).
 \label{eq:angular_calibration}
\end{equation}
Thus its angular coefficient matches the angular Fisher information.
The radial coefficient generally differs from $B(\kappa)$.
This is an angular-information calibration. Sampling uncertainty is
propagated with the empirical sandwich covariance in
Section~\ref{sec:sandwich}, whose von Mises limit is determined by
\eqref{eq:fisher_information}. The pullback calculation is given in
Appendix~\ref{app:metric}.

\subsection{A score relative to a known preferred direction}
\label{sec:score}
Let $\xi_0=e^{i\mu_0}$ denote the preferred boundary point.
It specifies a direction without specifying a target population.
For $z\in\mathbb D$, use the normalized Busemann function
\begin{equation}
 \beta_{\xi_0}(z)
 =\log\frac{|\xi_0-z|^2}{1-|z|^2},\qquad \beta_{\xi_0}(0)=0.
 \label{eq:busemann_disk}
\end{equation}
This is a signed score, with smaller values favoured relative to
$\xi_0$. It is finite throughout the disk. Its geometric definition
as a renormalized distance and its closed form are standard \citep{ghadimi2021busemann}; the derivation needed here appears in
Appendix~\ref{app:busemann}. A level set of this score is called a
\emph{horocycle}: in the disk it is a Euclidean circle tangent to the
boundary at $\xi_0$, with the tangency point excluded. Lower score
levels bound smaller regions toward $\xi_0$
(Figure~\ref{fig:horocycles_method}a). The circle equation is derived
in Appendix~\ref{app:horocycle_geometry}. Composing \eqref{eq:busemann_disk} with \eqref{eq:embedding} gives
\begin{equation}
 \begin{aligned}
 \beta(\mu,\kappa;\mu_0)&=\log D(\phi,\kappa),\\
 D(\phi,\kappa)&=
 \sqrt{1+C(\kappa)}-\sqrt{C(\kappa)}\cos\phi.
 \end{aligned}
 \label{eq:score}
\end{equation}
The expression inside the logarithm is strictly positive.
The score is smooth for every $\kappa>0$ and every direction, including
$\phi=0$ and $\phi=\pi$. It is unchanged by a joint rotation of the
population and target and by reflection about the target axis.

Put
\begin{equation}
 \rho(\kappa)=\asinh\sqrt{C(\kappa)}.
 \label{eq:rho}
\end{equation}
At a fixed concentration, the score increases strictly with absolute
departure on $(0,\pi)$ and ranges from $-\rho(\kappa)$ at alignment to
$\rho(\kappa)$ at opposition. Its continuous limit at uniformity is zero.

\subsection{Population contrast and the meaning of equality}
\label{sec:contrast}
Define
\begin{equation}
 \Delta_\beta=
 \beta(\mu_1,\kappa_1;\mu_0)
 -\beta(\mu_2,\kappa_2;\mu_0),\qquad
 H_0:\Delta_\beta=0\quad\text{versus}\quad H_1:\Delta_\beta\ne0.
 \label{eq:hypotheses}
\end{equation}
A negative value favours group 1 under the chosen score.
The two population points lie on the same horocycle if and only if
$\Delta_\beta=0$. For example, with $\mu_0=0$, the point for
$\operatorname{VM}(20^\circ,1.5)$ and its positive inverse-score
partner at concentration 3 share the level $\beta\approx-0.718$.
The partner's mean is approximately $28.93^\circ$
(Figure~\ref{fig:horocycles_method}b). The same score can therefore
arise from different directions and concentrations.

\begin{figure}[!htbp]
 \centering
 \includegraphics[width=.90\linewidth]{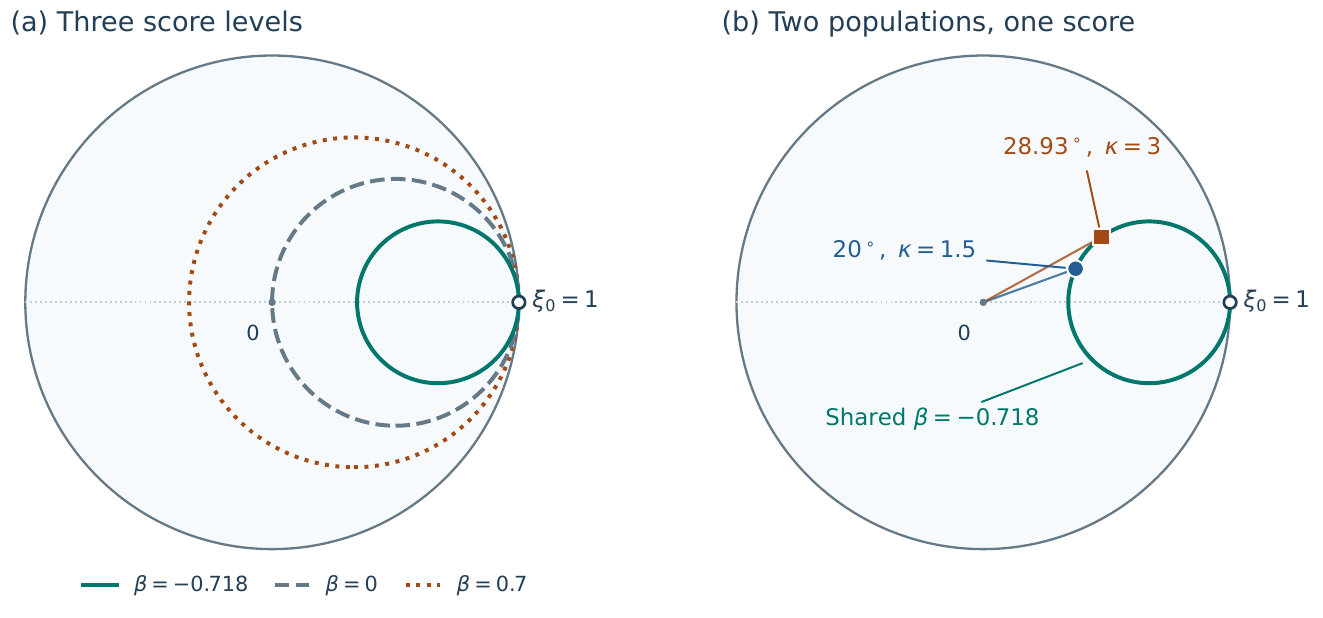}
 \caption{Busemann score levels in coordinates with $\mu_0=0$.
 (a) Horocycles at three score levels are tangent to the disk boundary
 at the preferred direction $\xi_0=1$. The smaller score corresponds
 to the smaller circle. (b) The two population points are on the same
 level $c\approx-0.718$. Radial segments show their different
 directions and concentrations; the second direction is computed at
 the exact shared score before rounding.}
 \label{fig:horocycles_method}
\end{figure}

Concentration is part of this estimand and contributes to its uncertainty.
Equal absolute angular departures generally do not imply
\eqref{eq:hypotheses} when the concentrations differ. Conversely, score
equality can balance a larger mean departure against a different
concentration. The score therefore defines a particular population
ordering. It is neither the expected cosine loss in
\eqref{eq:intro_cosine_loss} nor a metric between distributions.
Appendix~\ref{app:null_construction} gives an analytic description of its
null contours, including a useful inverse formula.

\FloatBarrier
\section{Asymptotic inference with unrestricted concentrations}
\label{sec:inference}
\subsection{Separate likelihood estimates}
\label{sec:estimation}
For each group, define the empirical first circular moment by
\begin{equation}
 \widehat m_g=\begin{pmatrix}\bar c_g\\\bar s_g\end{pmatrix}
 =\frac1{n_g}\sum_{j=1}^{n_g}
 \begin{pmatrix}\cos\Theta_{gj}\\\sin\Theta_{gj}\end{pmatrix},
 \qquad \bar R_g=\|\widehat m_g\|.
 \label{eq:sample_moment}
\end{equation}
The separate maximum likelihood estimates are
\begin{equation}
 \widehat\mu_g=\Arg(\bar c_g+i\bar s_g)\pmod{2\pi},\qquad
 \widehat\kappa_g=A^{-1}(\bar R_g).
 \label{eq:mle}
\end{equation}
For $0<\bar R_g<1$, the concentration estimate is finite, positive and
unique. Estimate the population contrast in \eqref{eq:hypotheses} by
\begin{equation}
 \widehat\Delta_\beta=
 \beta(\widehat\mu_1,\widehat\kappa_1;\mu_0)
 -\beta(\widehat\mu_2,\widehat\kappa_2;\mu_0).
 \label{eq:estimate}
\end{equation}
The point estimates remain likelihood estimates; the sandwich construction
below estimates their covariance.

\subsection{Score derivatives and first-order regularity}
\label{sec:score_derivatives}
A regular point of the score null has a nonzero contrast gradient and
positive first-order variance. A regular branch is a smooth local part
of that null surface (Appendix~\ref{app:positive_proof}).
Write $a(\mu,\kappa;\mu_0)=\partial\beta/\partial\mu$ and
$b(\mu,\kappa;\mu_0)=\partial\beta/\partial\kappa$. Then
\begin{equation}
 \begin{aligned}
 a(\mu,\kappa;\mu_0)
 &=\frac{\sqrt{C(\kappa)}\sin\phi}{D(\phi,\kappa)},\\
 b(\mu,\kappa;\mu_0)
 &=\frac{C'(\kappa)}{2D(\phi,\kappa)}
 \left\{\frac1{\sqrt{1+C(\kappa)}}
       -\frac{\cos\phi}{\sqrt{C(\kappa)}}\right\},\\
 C'(\kappa)&=A(\kappa)+\kappa B(\kappa).
 \end{aligned}
 \label{eq:gradient}
\end{equation}
Under the von Mises model, the population variance factor is
\begin{equation}
 v_\beta(\mu,\kappa;\mu_0)=
 \frac{a(\mu,\kappa;\mu_0)^2}{C(\kappa)}
 +\frac{b(\mu,\kappa;\mu_0)^2}{B(\kappa)}.
 \label{eq:variance_factor}
\end{equation}
This formula characterizes the limiting variance; the implemented standard
error uses its empirical sandwich counterpart below.

\begin{proposition}[Positive first-order variance]
\label{prop:positive_variance}
For every $\kappa>0$ and every $\mu$, $v_\beta(\mu,\kappa;\mu_0)>0$.
At alignment and opposition, respectively,
\begin{equation}
 (a,b)=(0,-\rho'(\kappa))\quad\hbox{and}\quad
 (a,b)=(0,\rho'(\kappa)),\qquad
 v_\beta=\frac{\rho'(\kappa)^2}{B(\kappa)}>0,
 \label{eq:endpoint_variance}
\end{equation}
where $\rho'(\kappa)=C'(\kappa)/[2\sqrt{C(\kappa)\{1+C(\kappa)\}}]$.
\end{proposition}
The proof is in Appendix~\ref{app:positive_proof}. Separate concentration
estimation supplies a nonzero linear contribution at the angular endpoints.

\subsection{Empirical sandwich standard error}
\label{sec:sandwich}
Let $X_{gj}=(\cos\Theta_{gj},\sin\Theta_{gj})^\top$ and use the
empirical per-observation covariance
\begin{equation}
 \widehat\Sigma_g=\frac1{n_g}\sum_{j=1}^{n_g}
 (X_{gj}-\widehat m_g)(X_{gj}-\widehat m_g)^\top.
 \label{eq:empirical_covariance}
\end{equation}
The divisor is $n_g$. No diagonal approximation is imposed. Define
\[
 \widehat u_{\parallel g}=\frac{\widehat m_g}{\bar R_g},\qquad
 \widehat u_{\perp g}=\frac{(-\bar s_g,\bar c_g)^\top}{\bar R_g},
 \qquad
 \widehat a_g=a(\widehat\mu_g,\widehat\kappa_g;\mu_0),\quad
 \widehat b_g=b(\widehat\mu_g,\widehat\kappa_g;\mu_0).
\]
The score gradient with respect to the first moment and the resulting
contrast variance estimate are
\begin{equation}
 \begin{aligned}
 \widehat q_g&=\frac{\widehat a_g}{\bar R_g}\widehat u_{\perp g}
 +\frac{\widehat b_g}{B(\widehat\kappa_g)}\widehat u_{\parallel g},\\
 \widehat s_\beta^2&=
 \frac{\widehat q_1^\top\widehat\Sigma_1\widehat q_1}{n_1}
 +\frac{\widehat q_2^\top\widehat\Sigma_2\widehat q_2}{n_2}.
 \end{aligned}
 \label{eq:standard_error}
\end{equation}
This retains empirical angular--concentration covariance within each group.
Appendix~\ref{app:sandwich_equivalence} shows its equivalence to the usual
likelihood-score sandwich. Under the von Mises model it consistently
estimates the variance determined by \eqref{eq:variance_factor}.

\subsection{Null distribution, test and confidence interval}
\label{sec:test}
Let $N=n_1+n_2$ and assume $n_1/N\to\lambda\in(0,1)$.
\begin{theorem}[Sandwich-based separate-fit inference]
\label{thm:normal}
Under \eqref{eq:model}, at any fixed parameter pair with
$\kappa_1,\kappa_2>0$,
\begin{equation}
 \frac{\widehat\Delta_\beta-\Delta_\beta}{\widehat s_\beta}
 \Rightarrow N(0,1).
 \label{eq:normal_limit}
\end{equation}
Consequently, at every fixed point of $H_0:\Delta_\beta=0$,
$Z_\beta=\widehat\Delta_\beta/\widehat s_\beta\Rightarrow N(0,1)$.
The statement includes equal true concentrations and exact alignment or
opposition, because the two concentrations are still estimated separately.
\end{theorem}
Appendix~\ref{app:wald_proof} gives the proof using the moment central limit
theorem, the delta method and consistent studentization \citep{vaart1998asymptotic}. The null is composite, but a constrained
null fit is unnecessary: under score equality the unrestricted contrast is
centered at zero and its sandwich standard error is consistent.

For a two-sided nominal level $\alpha$, use
\begin{equation}
 \begin{aligned}
 p_\beta&=2\{1-\Phi(|Z_\beta|)\},\\
 \operatorname{CI}_{1-\alpha}(\Delta_\beta)
 &: \widehat\Delta_\beta\ \pm\ z_{1-\alpha/2}\widehat s_\beta.
 \end{aligned}
 \label{eq:test_interval}
\end{equation}
Reject when $|Z_\beta|>z_{1-\alpha/2}$. The test has pointwise asymptotic
level $\alpha$ and is consistent against every fixed nonzero score contrast.
An interval entirely below zero supports a smaller score for group 1.
An interval containing zero does not establish equivalence.

\subsection{Implementation and scope}
\label{sec:computation}
For stable score evaluation, use
\begin{equation}
 D(\phi,\kappa)=
 \frac1{\sqrt{1+C(\kappa)}+\sqrt{C(\kappa)}}
 +2\sqrt{C(\kappa)}\sin^2(\phi/2).
 \label{eq:stable_score}
\end{equation}
Scaled Bessel functions and a bracketed solution of
$A(\widehat\kappa_g)=\bar R_g$ avoid unnecessary numerical overflow.
Appendix~\ref{app:derivatives} gives a stable formula for $b$.
An implementation records a noninterior fit or a nonpositive computed
variance explicitly rather than substituting a variance floor.

The limit is pointwise at fixed positive concentrations; it does not assert
exact finite-sample calibration or uniform validity near zero concentration.
Appendix~\ref{app:moment_scope} specifies the moment-matched target under
independent non-von-Mises observations. Dependent data or an estimated
preferred direction require additional covariance terms. The local-power
derivation and its validation are in Appendices~\ref{app:local_power}
and~\ref{app:simulation_local}.

\section{Simulation study}
\label{sec:simulation}
\subsection{Design}
We generated independent von Mises samples with $\mu_0=0$ and balanced
groups at total sizes $N=40,100,200,500$. Each setting used 10,000
replications and a two-sided nominal level of 0.05. The primary procedure
(BS-US) estimates the concentrations separately and uses the empirical
sandwich in \eqref{eq:standard_error}. The reference angles
$\gamma_j=j\pi/20$, $j=0,\ldots,20$, run from $0^\circ$ to $180^\circ$
in $9^\circ$ steps.

The unequal-concentration score-null designs use ordered pairs
$(1.48,1.5)$, $(1.5,3)$ and $(3,1.5)$. The less concentrated population
uses $\gamma_j$; both partners for the other population are obtained
from \eqref{eq:inverse_score} without rounding. For the unequal-concentration power curve, we fixed
$\phi_1=0$ and used $(\kappa_1,\kappa_2)=(1.5,12)$. We set
$d_\beta=\beta_2-\beta_1$ on a grid from $-1$ to $1$ in steps of
$0.05$ and obtained each $\phi_2$ from \eqref{eq:inverse_score}.
With $\phi_1=0$, the feasible range is
$[\rho(\kappa_1)-\rho(\kappa_2),\,\rho(\kappa_1)+\rho(\kappa_2)]$.
The pair $(1.5,12)$ makes the entire grid feasible; it is used
for the power curve separately from the score-null grids.
For common concentrations $\kappa=1,1.5,3$, the same-side and reflected
pairs $(\gamma_j,\pm\gamma_j)$ give a shared null for BS-US and the cosine
test of \citet{biswas2016comparison} (Biswas $Z$).
Appendix~\ref{app:simulation} describes the settings and diagnostics.

\subsection{Calibration and power under the score criterion}
Table~\ref{tab:sim_main_level} reports the range of null rejection rates
across each grid.
At $N=40$, BS-US rejection reached $7.83\%$ under unequal concentrations and $8.60\%$ under common concentration. At $N=500$, BS-US rates on the two grids were between $4.34\%$ and $5.68\%$. Biswas $Z$ stayed closer to the nominal level at $N=40$ on the common-concentration grid, as judged by the maximum absolute deviation.

\begin{table}[!htbp]
 \centering\small
 \caption{Ranges of null rejection percentages at nominal $5\%$.
 BS-US uses separate fits in both concentration regimes. The unequal
 score null is distinct from the Biswas $Z$ null. The ranges describe
 the simulated grid and are not simultaneous confidence bounds.}
 \label{tab:sim_main_level}
 \begin{tabular}{rccc}
\toprule
$N$ & Unequal: BS-US & Common: BS-US & Common: Biswas $Z$ \\
\midrule
40 & $[5.46,\,7.83]$ & $[4.95,\,8.60]$ & $[4.24,\,5.83]$ \\
100 & $[4.75,\,6.02]$ & $[4.93,\,6.68]$ & $[4.59,\,5.51]$ \\
200 & $[4.88,\,5.93]$ & $[4.70,\,5.78]$ & $[4.55,\,5.51]$ \\
500 & $[4.51,\,5.68]$ & $[4.34,\,5.54]$ & $[4.36,\,5.44]$ \\
\bottomrule
\end{tabular}

\end{table}

Figure~\ref{fig:sim_unequal} shows rejection along the signed
score alternatives with the aligned anchor. The central point
$d_\beta=0$ is the exact score null, and the remaining points are
alternatives.
At $N=200$, the estimated rejection probabilities at $d_\beta=-0.25$ and $0.25$ are $0.523$ and $0.552$, respectively. At the exact null the rejection rate is $5.32\%$.

The unequal-concentration Biswas benchmarks test different criteria.
Appendix~\ref{app:simulation_targets} reports their calibration under
their own nulls and their rejection rates on shared configurations.

\begin{figure}[!htbp]
 \centering
 \includegraphics[width=.88\linewidth]{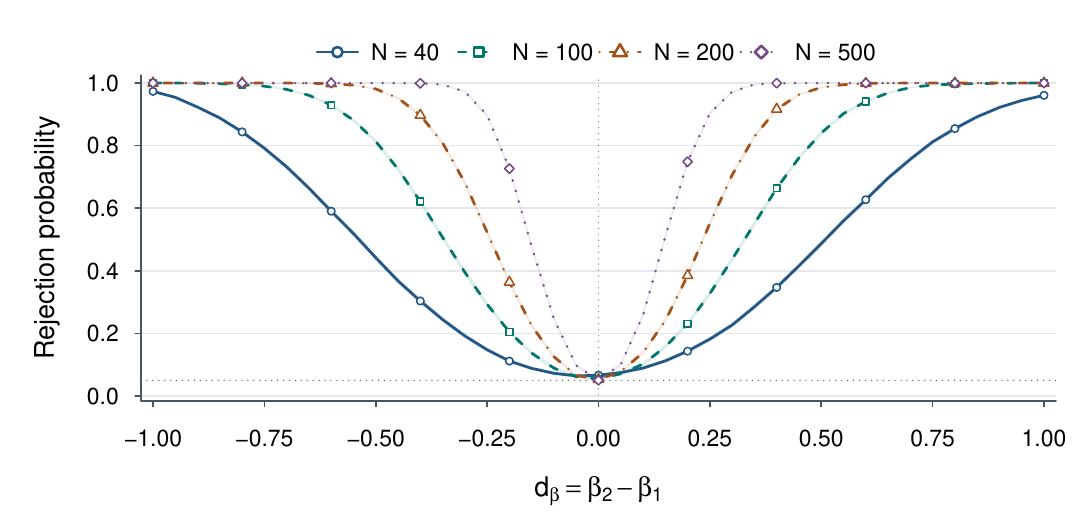}
 \caption{BS-US rejection probabilities for $d_\beta=\beta_2-\beta_1$
 with $\phi_1=0$ and $(\kappa_1,\kappa_2)=(1.5,12)$. The vertical
 reference marks the exact score null at $d_\beta=0$. Shading gives
 pointwise 95\% Monte Carlo Wilson intervals. The horizontal reference
 is 0.05, and $N$ is the total sample size.}
 \label{fig:sim_unequal}
\end{figure}

\subsection{Common-concentration power comparison}
We follow the alignment benchmark $\phi_1=0$, $\phi_2=\gamma_j$.
Figure~\ref{fig:sim_common} compares BS-US and Biswas $Z$ for
$\kappa=1.5$ at $N=200$ and $500$, using nominal critical values.
At $N=200$, shifting the second mean from $0^\circ$ to $27^\circ$
gives estimated power $0.424$ for BS-US and $0.150$ for Biswas $Z$.
This comparison concerns the specified alignment path.

\begin{figure}[!htbp]
 \centering
 \includegraphics[width=.85\linewidth]{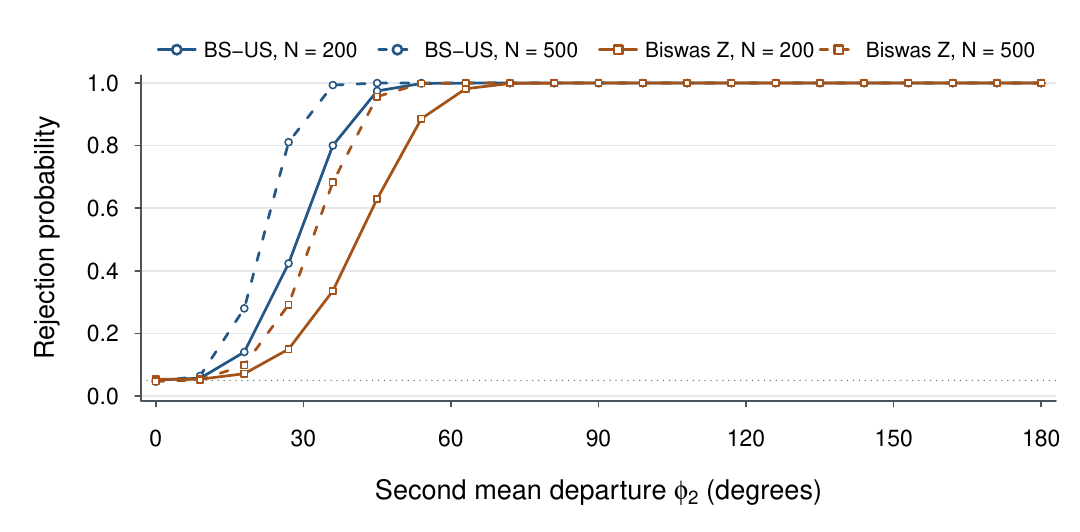}
 \caption{Rejection probabilities on the aligned common-concentration
 benchmark, with $\kappa=1.5$ and $\phi_1=0$. The shared null is
 $\phi_2=0$. Markers identify procedures, and line styles distinguish
 total sample sizes. The horizontal reference is 0.05.}
 \label{fig:sim_common}
\end{figure}

\FloatBarrier
\section{Cataract-axis data analysis}
\label{sec:data_analysis}
\subsection{Data and analysis target}
The clinical study is documented in the thesis of \citet{bakshi2010thesis}
and described by \citet{biswas2016comparison}.
We analyze deidentified angular extracts for the SNARE and irrigating
VECTIS procedures.
For the three-month postoperative axis endpoint, we take the
\texttt{AX3\_manual\_rad4\_snare} and
\texttt{AX3\_manual\_rad4\_vec} columns as primary.
These columns already contain the fourfold angles in radians.
We use them directly, with preferred direction $\mu_0=0$ as explained
in Section~\ref{sec:preferred_coding}. The analysis compares the
transformed postoperative axes.

The primary analysis retains all nonmissing values: 20 for SNARE and
19 for VECTIS. One VECTIS row is entirely missing. Rows are treated as
independent analysis units; participant and eye identifiers were not
available in the extracts. Other measurement columns are not counted as
additional independent observations. One entry in each primary column
coincides with its circular mean. We retain these entries in the
primary analysis and examine their exclusion as a sensitivity analysis.
Appendix~\ref{app:data_details} gives the extraction audit and the
distinction between the supplied record counts and the historical
clinical example.

The primary estimand is
$\Delta_\beta=\beta_{\rm SNARE}-\beta_{\rm VECTIS}$, estimated with
separate concentration fits and the full empirical sandwich covariance
in Section~\ref{sec:inference}. Negative
values favour SNARE under the score criterion.

\subsection{Fitted populations and score contrast}
Table~\ref{tab:cataract_busemann_fits} reports the separate von Mises
fits. The fitted mean directions are $0.3066$ and $0.5402$ transformed
radians, with concentrations $1.5675$ and $1.5887$. The corresponding
Busemann scores are $-0.7637$ for SNARE and $-0.5832$ for VECTIS.
Figure~\ref{fig:cataract_horocycles} places the fits on their score
contours. The SNARE point lies on the smaller horocycle, so its
estimated score is lower. The plot describes the fitted ordering;
uncertainty in that ordering is assessed by the interval below.

\begin{table}[!htbp]
 \centering\small
 \caption{Separate von Mises fits for the supplied primary AX3
 manual-axis endpoint. Mean directions are in fourfold-transformed
 radians. Scores and radii use the representation in Section~\ref{sec:embedding}.}
 \label{tab:cataract_busemann_fits}
 \begin{tabular}{lrrrrrr}
\toprule
Group & $n$ & $\widehat\mu$ & $\bar R$ & $\widehat\kappa$ & $R_F(\widehat\kappa)$ & $\widehat\beta$ \\
\midrule
SNARE & 20 & 0.3066 & 0.6124 & 1.5675 & 0.4082 & -0.7637 \\
VECTIS & 19 & 0.5402 & 0.6173 & 1.5887 & 0.4114 & -0.5832 \\
\bottomrule
\end{tabular}

\end{table}

\begin{figure}[!htbp]
 \centering
 \includegraphics[width=.93\linewidth]{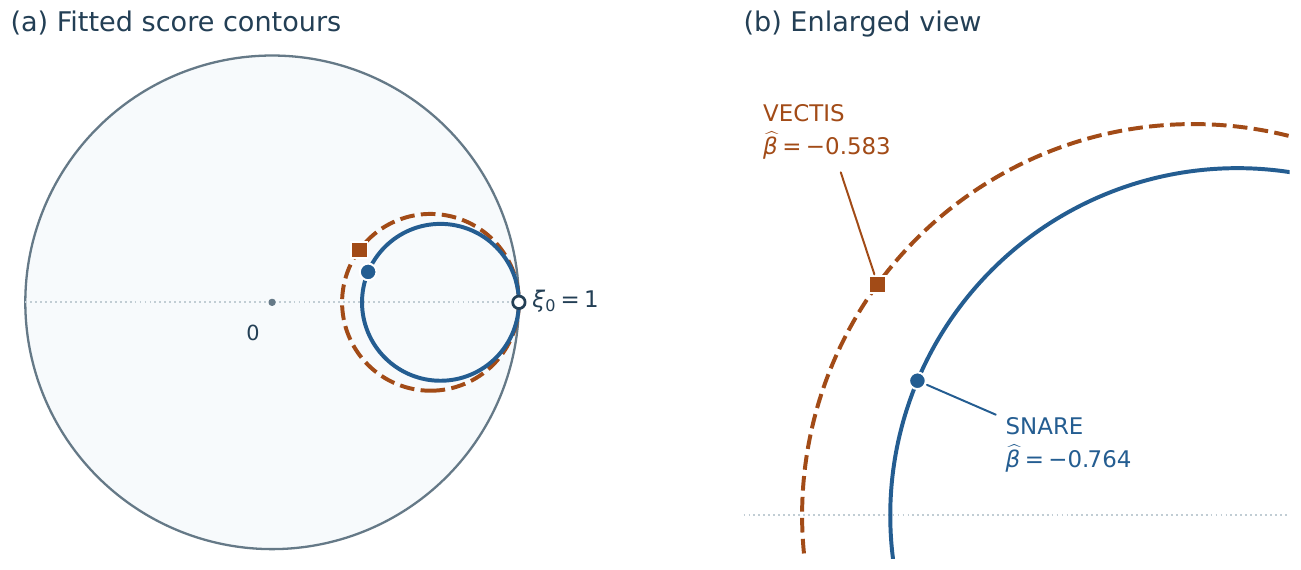}
 \caption{Horocycles through the separate fits for the primary
 cataract-axis endpoint. (a) The two contours touch the disk boundary
 at $\xi_0=1$, the preferred transformed direction.
 (b) An enlarged view separates the fitted locations.
 The contours represent point estimates, with levels
 $\widehat\beta_{\rm SNARE}=-0.764$ and
 $\widehat\beta_{\rm VECTIS}=-0.583$.}
 \label{fig:cataract_horocycles}
\end{figure}

The estimated contrast is $\widehat\Delta_\beta=-0.1805$, with sandwich standard error $0.2895$. The two-sided normal-reference $p$-value is $0.5328$, and the nominal 95\% interval is $[-0.7479,\,0.3868]$ (Table~\ref{tab:cataract_busemann_results}). The point estimate favours SNARE under the stated score, but the interval includes zero.

\begin{table}[!htbp]
 \centering\small
 \setlength{\tabcolsep}{5pt}
 \caption{Primary inference, with SNARE as group 1. The proposed
 estimate, sandwich standard error and interval concern $\Delta_\beta$.
 Biswas $W$ concerns equality of absolute angular departures and
 uses its original chi-square reference; its statistic is on a
 different scale. Both reported tests are two-sided.}
 \label{tab:cataract_busemann_results}
 \begin{tabular}{lrrrrl}
\toprule
Method & Estimate & SE & Statistic & $p$-value & 95\% score interval \\
\midrule
Proposed score & -0.1805 & 0.2895 & -0.6237 & 0.5328 & $[-0.7479,\ 0.3868]$ \\
Biswas $W$ & -- & -- & 0.5160 & 0.4726 & -- \\
\bottomrule
\end{tabular}

\end{table}

\subsection{Comparator, sensitivity and interpretation}
The original angular procedure of \citet{biswas2016comparison}
gives $W=0.5160$ and $p=0.4726$. Its formulas are retained in
Appendix~\ref{app:data_comparator}. The score and angular procedures
both fail to reject their respective equality nulls at 0.05. Their
different $p$-values do not measure relative power for this dataset,
because the unrestricted score null and the angular null generally
differ.
Excluding the two mean-equal entries gives 19 SNARE and 18 VECTIS values. The proposed contrast is $-0.1698$, with sandwich standard error $0.2983$, nominal 95\% interval $[-0.7544,\,0.4148]$ and $p=0.5691$. The original angular procedure gives $p=0.5038$.

This sensitivity analysis preserves the estimated ordering and the
nonrejection conclusion. It addresses uncertainty about those supplied
entries; it does not establish their imputation status.
The intervals remain wide and include zero. Nonrejection should not
be interpreted as treatment equivalence or equality of distributions.
The samples contain only 20 and 19 records. The normal-reference
$p$-values and the score interval are approximate: at similar group
sizes, the simulations found liberal rejection at some parameter
settings (Table~\ref{tab:sim_main_level}). Those simulations do not
calibrate this particular dataset. The fitted scores describe the
supplied axis endpoint; a clinical ranking would require a
substantively justified outcome criterion and additional information.

\FloatBarrier
\section{Discussion}
\label{sec:discussion}
The framework compares circular populations through a smooth score relative
to a known direction. The representation is calibrated to angular Fisher
information, and the score combines directional alignment and concentration.
Separate likelihood estimates with an empirical sandwich covariance account
for uncertainty in both parameters. Under the von Mises model, the
first-order variance is positive at every fixed positive-concentration pair,
including alignment and opposition.

The estimand determines the interpretation. At a fixed concentration, a
smaller absolute departure gives a smaller score. With different
concentrations, the score also reflects the interaction between concentration
and departure. Its ordering need not agree with an ordering based on absolute
mean departure or expected individual cosine loss. A scientifically useful
application therefore specifies the preferred direction, angular coding and
population criterion before comparing groups.

The cataract-axis illustration yields a smaller estimated score for SNARE,
but its interval includes zero. Excluding the two mean-equal entries gives the
same qualitative interpretation. These are comparisons of the
transformed-axis extracts under an independent-observation analysis; they
do not establish treatment equivalence or clinical superiority. The Biswas
angular comparator addresses a distinct null in the unrestricted model.
Differences between observed $p$-values cannot establish relative power.

The geometric construction is not a variance-stabilizing transformation.
The sandwich standard error is part of the inferential procedure, and its
asymptotic justification is pointwise. 

The simulations show that finite-sample calibration can remain liberal
at the smallest size: nominal 5\% rejection reaches 8.60\% at $N=40$.
The core score-null rejection ranges are closer to 5\% at $N=500$.
Under common concentration, power relative to Biswas $Z$ depends on the
alternative. Under unequal concentrations, score and angular comparisons
generally test different hypotheses.

Under model misspecification the moment
argument has a specified moment-matched target; it does not turn arbitrary
circular populations into von Mises distributions. Estimated preferred
directions and clustered or longitudinal observations require corresponding
extensions of the covariance calculation.

\section*{Acknowledgments}
Daitari Prusty acknowledges support from the Government of India through the Joint Council of Scientific and Industrial Research University Grants Commission National Eligibility Test, NTA Reference Number 241610013744.

\clearpage
\bibliographystyle{preferred_author_year}
\bibliography{references}
\clearpage
\appendix
\section{Geometry and inverse scores}
\label{app:geometry}
\subsection{Information and calibration}
\label{app:information}
Let $U=\Theta-\mu$. The scores are $\kappa\sin U$ and
$\cos U-A(\kappa)$. Symmetry and the Bessel recurrences \citep{nistdlmf} give
\[
 E(\sin U)=\Cov(\sin U,\cos U)=0,\quad
 E(\sin^2U)=A(\kappa)/\kappa,\quad
 \Var(\cos U)=B(\kappa)>0.
\]
This proves \eqref{eq:fisher_information}. Also
$C'(\kappa)=\kappa\{1-A(\kappa)^2\}>0$.
Since $A(\kappa)\sim\kappa/2$ at zero and $A(\kappa)\to1$ at infinity,
$C(\kappa)$ ranges from zero to infinity.

\label{app:metric}
The angular metric coefficient $4r^2/(1-r^2)^2=C(\kappa)$ gives
\eqref{eq:embedding}. The radial coordinate is
$\rho(\kappa)=2\atanh R_F(\kappa)$, so the pullback is
$C(\kappa)d\mu^2+\rho'(\kappa)^2d\kappa^2$ with
$\rho'(\kappa)=C'(\kappa)/[2\sqrt{C(\kappa)\{1+C(\kappa)\}}]$.
The radial coefficient is not generally $B(\kappa)$.

\subsection{Score and its level sets}
\label{app:busemann}
Along $\gamma(t)=\xi_0\tanh(t/2)$ the disk distance is
\[
 d_H(z,\gamma(t))=\arcosh\!\left(
  1+\frac{2|z-\gamma(t)|^2}{(1-|z|^2)(1-|\gamma(t)|^2)}\right).
\]
Using $1-|\gamma(t)|^2\sim4e^{-t}$ and
$\arcosh(1+x)=\log(2x)+o(1)$ gives \eqref{eq:busemann_disk}
as $\lim_{t\to\infty}\{d_H(z,\gamma(t))-t\}$ \citep{ghadimi2021busemann}.
For $r=R_F(\kappa)$, the identities
$2r/(1-r^2)=\sqrt{C(\kappa)}$ and
$(1+r^2)/(1-r^2)=\sqrt{1+C(\kappa)}$ give \eqref{eq:score},
or $\beta=\log\{\cosh\rho(\kappa)-\sinh\rho(\kappa)\cos\phi\}$.

\label{app:horocycle_geometry}
Write $x+iy=e^{-i\mu_0}z$ and $q=e^c$.
Completing the square in $(1-x)^2+y^2=q(1-x^2-y^2)$ gives
\[
 \left(x-\frac1{1+q}\right)^2+y^2=\left(\frac q{1+q}\right)^2,
 \qquad |z|<1 .
\]
The circles bound nested disks and meet only at the excluded point $(1,0)$.

\subsection{Computation}
\label{app:derivatives}
Writing $d_0=e^{-\rho(\kappa)}$, differentiate \eqref{eq:stable_score} for
a stable radial derivative:
\[
 b=\rho'(\kappa)
 \frac{-d_0+2\sqrt{1+C(\kappa)}\sin^2(\phi/2)}
      {d_0+2\sqrt{C(\kappa)}\sin^2(\phi/2)} .
\]
Differentiating the radial score also yields
$a^2/C(\kappa)+b^2/\rho'(\kappa)^2=1$;
the sampling variance in \eqref{eq:variance_factor} uses $B(\kappa)$ instead.

\label{app:null_construction}
Solving $D(\phi,\kappa)=e^c$ gives both branches:
\begin{equation}
 \phi=\pm\arccos\left\{\frac{\sqrt{1+C(\kappa)}-e^c}{\sqrt{C(\kappa)}}\right\}
 \pmod{2\pi},\qquad |c|\le \rho(\kappa).
 \label{eq:inverse_score}
\end{equation}
A common score must belong to both groups' feasible score ranges.

\section{Proofs and local power}
\label{app:proofs}
\subsection{Moment map and covariance}
\label{app:sandwich_equivalence}
Let $m_g=A(\kappa_g)u_{\parallel g}$, with $u_{\parallel g}$ the unit
vector at $\mu_g$ and $u_{\perp g}$ its counterclockwise normal.
The Jacobian $T_g$ of the moment map has rows
$u_{\perp g}^{\top}/A(\kappa_g)$ and $u_{\parallel g}^{\top}/B(\kappa_g)$.
Thus $q_g=T_g^\top(a_g,b_g)^\top$ gives the influence function
\[
 \psi_{\beta,g}(\theta)=
 \frac{a_g}{A(\kappa_g)}\sin(\theta-\mu_g)
 +\frac{b_g}{B(\kappa_g)}
   \{\cos(\theta-\mu_g)-A(\kappa_g)\}.
\]
In local angular coordinates the score error is
$n_g^{-1}\sum_j\psi_{\beta,g}(\Theta_{gj})+o_p(n_g^{-1/2})$.
At the MLE the average negative Hessian is
$\widehat H_g=\diag\{C(\widehat\kappa_g),B(\widehat\kappa_g)\}$:
the fitted sine moment vanishes and the cosine moment is
$A(\widehat\kappa_g)$. Observation by observation,
$\widehat H_g^{-1}s_g(\theta;\widehat\eta_g)
 =\widehat T_g\{(\cos\theta,\sin\theta)^\top-\widehat m_g\}$.
Hence the likelihood sandwich equals
$\widehat T_g\widehat\Sigma_g\widehat T_g^\top/n_g$,
including the off-diagonal empirical covariance.
Under von Mises sampling the radial and angular moment variances are
$B(\kappa_g)$ and $A(\kappa_g)/\kappa_g$, with covariance zero.
It follows that $q_g^\top\Sigma_g q_g=v_\beta(\mu_g,\kappa_g;\mu_0)$.

\subsection{Regularity and the test}
\label{app:positive_proof}
If $\sin\phi\ne0$, then $a\ne0$; otherwise $b=\mp\rho'(\kappa)\ne0$.
This proves Proposition~\ref{prop:positive_variance}. The unrestricted
gradient $(a_1,b_1,-a_2,-b_2)^\top$ never vanishes, so the implicit
function theorem gives a smooth null hypersurface locally.
A regular point has nonzero gradient and positive first-order variance;
a regular branch is a local smooth part of that hypersurface.

\label{app:wald_proof}
The independent influence sums give
$\sqrt N(\widehat\Delta_\beta-\Delta_\beta)\Rightarrow N(0,V_\beta)$,
where $V_\beta=q_1^\top\Sigma_1q_1/\lambda+
q_2^\top\Sigma_2q_2/(1-\lambda)>0$.
The law of large numbers and continuity give
$N\widehat s_\beta^2\xrightarrow{p}V_\beta$. Slutsky's theorem gives
Theorem~\ref{thm:normal}, pointwise size and interval coverage;
the test is consistent for each fixed nonzero contrast.

\subsection{Local alternatives and moment-matched scope}
\label{app:local_power}
At a fixed null $\vartheta_0$, let
$\vartheta_N=\vartheta_0+h/\sqrt N+o(N^{-1/2})$.
The triangular-array moment central limit theorem applies because
the moments are bounded and their covariances converge.
The local mean shift of $Z_\beta$ is
$\nu=(a_1h_{\mu1}+b_1h_{\kappa1}-a_2h_{\mu2}-b_2h_{\kappa2})/
\sqrt{V_\beta(\vartheta_0)}$, with derivatives evaluated at the null.
Limiting two-sided power is
\begin{equation}
 1-\Phi(z_{1-\alpha/2}-\nu)+\Phi(-z_{1-\alpha/2}-\nu).
 \label{eq:local_power}
\end{equation}
At a joint endpoint, a purely angular perturbation has zero first-order
shift despite positive variance. Appendix~\ref{app:simulation_local}
checks the formula at regular alternatives.

\label{app:moment_scope}
For independent non-von-Mises angles with $0<\|E(X)\|<1$ and positive
definite covariance, the same moment argument applies to
$(\mu^*,\kappa^*)=(\Arg E(e^{i\Theta}),A^{-1}(\|E(X)\|))$.
The inferred score then concerns this moment-matched population.

\section{Common-concentration specialization}
\label{app:common}
\subsection{Restricted fit and joint covariance}
Under $\kappa_1=\kappa_2=\kappa$, equality of scores is equivalent to
$\cos\phi_1=\cos\phi_2$ and to $|\phi_1|=|\phi_2|$ on principal angles.
Put $w_g=n_g/N$ and $\widehat\kappa_P=A^{-1}(\sum_gw_g\bar R_g)$.
The contrast $\widehat\Delta_{\beta,P}$ evaluates both fitted scores at
$\widehat\kappa_P$. Write $\widehat a_{g,P},\widehat b_{g,P}$ for the
derivatives at $(\widehat\mu_g,\widehat\kappa_P)$. Differentiation gives
\begin{equation}
 \begin{aligned}
 \widehat q_{g,P}&=(-1)^{g-1}
   \frac{\widehat a_{g,P}}{\bar R_g}\widehat u_{\perp g}
 +\frac{w_g(\widehat b_{1,P}-\widehat b_{2,P})}
 {B(\widehat\kappa_P)}\widehat u_{\parallel g},\\
 \widehat s_{\beta,P}^2&=\sum_{g=1}^2
 \widehat q_{g,P}^{\top}\widehat\Sigma_g\widehat q_{g,P}/n_g .
 \end{aligned}
 \label{eq:pooled_variance}
\end{equation}
The shared radial estimate produces the difference in radial derivatives.
The model variance is
\[
 s_{P,0}^2=\frac{a_1^2}{n_1C(\kappa)}+\frac{a_2^2}{n_2C(\kappa)}
          +\frac{(b_1-b_2)^2}{NB(\kappa)}.
\]
Where this is positive, the same moment and delta argument gives a
standard normal centered pivot. Relative to the separate-fit model
variance, the first-order saving is
$(n_2b_1+n_1b_2)^2/\{n_1n_2NB(\kappa)\}\ge0$.
With unequal concentrations, the common fit instead converges to
$A^{-1}\{\lambda A(\kappa_1)+(1-\lambda)A(\kappa_2)\}$; its score
equality is then an angular null. We do not select pooling by a pretest.

\subsection{Alignment and opposition}
\label{app:pooled_endpoints}
At a joint endpoint $a_1=a_2=0$ and $b_1=b_2$, so the restricted
gradient vanishes. For $s=1$ at alignment and $s=-1$ at opposition,
let $q_s=s\sqrt{C(\kappa)}/\{\sqrt{1+C(\kappa)}-s\sqrt{C(\kappa)}\}$.
For fitted angular errors $\eta_1,\eta_2$, Taylor expansion gives
$\widehat\Delta_{\beta,P}=q_s(\eta_1^2-\eta_2^2)/2+o_p(N^{-1})$ and
$\widehat a_{g,P}=q_s\eta_g+o_p(N^{-1/2})$; the radial difference is
$O_p(N^{-1})$. Thus for independent $U_1,U_2\sim N(0,1)$,
\begin{equation}
 \begin{aligned}
 N\widehat\Delta_{\beta,P}
 &\Rightarrow \frac{q_s}{2C(\kappa)}
       \left\{\frac{U_1^2}{\lambda}-\frac{U_2^2}{1-\lambda}\right\},\\
 \frac{\widehat\Delta_{\beta,P}}{\widehat s_{\beta,P}}
 &\Rightarrow
 \frac{\sgn(q_s)\{U_1^2/\lambda-U_2^2/(1-\lambda)\}}
 {2\sqrt{U_1^2/\lambda^2+U_2^2/(1-\lambda)^2}}.
 \end{aligned}
 \label{eq:pooled_endpoint_studentization}
\end{equation}
The normal reference is inappropriate: at $N=500$ pooled rejection
was 0.00\%--0.02\% at joint endpoints and 0.56\%--5.38\% at interior
angles. The unrestricted pivot remains regular at the endpoints.

\subsection{Cosine benchmark and conditional application}
Let $\bar X_g$ and $\bar Y_g$ be means of $\cos(\Theta-\mu_0)$ and
$\cos\{2(\Theta-\mu_0)\}$, respectively, and use subscript $P$ for
their $w_g$-weighted means. The Biswas $Z$ statistic \citep{biswas2016comparison} uses
\begin{equation}
 Z_B=\frac{\bar X_1-\bar X_2}
 {\sqrt{\{(1+\bar Y_P-2\bar X_P^2)/2\}(n_1^{-1}+n_2^{-1})}}.
 \label{eq:biswas_Z}
\end{equation}
At a common concentration, the cosine-mean null is the score null.

For the primary endpoint, $\widehat\kappa_P=1.5778$,
$\widehat\Delta_{\beta,P}=-0.1848$ (SE 0.2600; 95\% interval
$[-0.6945,0.3249]$), $p_{\beta,P}=0.4773$ and $p_Z=0.7407$.
The concentration likelihood-ratio diagnostic ($p=0.9746$) does
not select the model; the separate-fit analysis remains primary.

\begin{table}[!htbp]
 \centering\small
 \setlength{\tabcolsep}{4pt}
 \caption{Conditional analysis under the common-concentration model
 for the primary endpoint. The pooled score uses the joint sandwich in
 \eqref{eq:pooled_variance}.}
 \label{tab:cataract_common_diagnostic}
 \begin{tabular}{lrrrrrr}
\toprule
Analysis & $\widehat\kappa_P$ & $\widehat\Delta_{\beta,P}$ & SE & $p_{\beta,P}$ & $Z_B$ & $p_Z$ \\
\midrule
As supplied & 1.5778 & -0.1848 & 0.2600 & 0.4773 & 0.3309 & 0.7407 \\
Mean-entry exclusion & 1.4912 & -0.1734 & 0.2667 & 0.5156 & 0.3109 & 0.7559 \\
\bottomrule
\end{tabular}

\end{table}

\section{Data definitions, audit and additional analyses}
\label{app:data_details}
\subsection{Data records and sensitivity analysis}
The extracts relate to the clinical study attributed to
\citet{bakshi2010thesis} by \citet{biswas2016comparison}.
They are the working angular data for this analysis; their row counts
are not a reconstruction of the original study records.
Each data file contains four angular columns and 20 source rows.
The primary columns are the AX3 manual-axis endpoints named in
Section~\ref{sec:data_analysis}. The other columns contain AX1 manual,
AX1 SIA and AX3 SIA endpoints. These are additional measurements rather
than independent replications of the primary endpoint.
We read the fourfold angles directly. The original axis
cannot be uniquely reconstructed from a many-to-one circular coding.

Table~\ref{tab:cataract_data_audit} records the extraction.
Only the missing VECTIS row is excluded from the primary analysis.
SNARE source row 6, with value $0.306641990856459$, and VECTIS source
row 15, with value $0.540245676704823$, each coincide with their
primary column's circular mean to within $10^{-12}$ radians.
Source rows are numbered from one after the CSV header.
Their numerical agreement does not establish how the values were obtained.

\begin{table}[!htbp]
 \centering\small
 \caption{Extraction audit of the supplied data. A mean-equal entry
 lies within $10^{-12}$ circular radians of the corresponding
 column's circular mean. This flags a numerical property rather
 than establishing how the value was obtained.}
 \label{tab:cataract_data_audit}
 \begin{tabular}{llrrrr}
\toprule
Group & Endpoint & Source rows & Missing & Used & Mean-equal row \\
\midrule
SNARE & AX3\_manual & 20 & 0 & 20 & 6 \\
VECTIS & AX3\_manual & 20 & 1 & 19 & 15 \\
SNARE & AX1\_manual & 20 & 0 & 20 & -- \\
VECTIS & AX1\_manual & 20 & 1 & 19 & -- \\
SNARE & AX1\_SIA & 20 & 0 & 20 & -- \\
VECTIS & AX1\_SIA & 20 & 1 & 19 & -- \\
SNARE & AX3\_SIA & 20 & 0 & 20 & 6 \\
VECTIS & AX3\_SIA & 20 & 1 & 19 & 15 \\
\bottomrule
\end{tabular}

\end{table}

Removing those two entries defines the 19/18 sensitivity analysis.
The source files are otherwise unchanged, and no replacement values
are introduced. \citet{biswas2016comparison,biswas2023alternative} report 19 and 18 observations.
Matching these counts does not establish that the sensitivity sample
reconstructs the original clinical records. The primary and sensitivity
results therefore retain the inclusion rules above.

\subsection{The original angular comparator}
\label{app:data_comparator}
For target-rotated observations, let
$\bar X_g=n_g^{-1}\sum_j\cos(\Theta_{gj}-\mu_0)$.
The Biswas $W$ statistic uses
\begin{equation}
 T_g=\arccos(\bar X_g/\bar R_g),\qquad
 S_g^2=\frac{1-I_2(\widehat\kappa_g)/I_0(\widehat\kappa_g)}
 {2n_g\bar R_g^2}.
 \label{eq:data_W_components}
\end{equation}
Its original quadratic expression and its equivalent stable form are
\begin{equation}
 W=\frac{T_1^2}{S_1^2}+\frac{T_2^2}{S_2^2}
 -\frac{(T_1/S_1^2+T_2/S_2^2)^2}{1/S_1^2+1/S_2^2}
 =\frac{(T_1-T_2)^2}{S_1^2+S_2^2}.
 \label{eq:data_W}
\end{equation}
We retain this procedure and the chi-square reference with one degree
of freedom \citep{biswas2016comparison}.
It compares absolute angular departures. In contrast, the proposed
procedure compares the concentration-dependent scores in
\eqref{eq:score}. These distinct population nulls explain why the
observed ordering of $p$-values cannot establish a power advantage.
The chi-square law for $W$ is itself a large-sample approximation at
regular angular parameters.

\subsection{Additional endpoints and fitted-model diagnostics}
Table~\ref{tab:cataract_endpoint_sensitivity} records the primary
analysis, its defined exclusion sensitivity and all three additional
endpoints. The latter are exploratory analyses. They are not treated
as independent confirmatory tests, and the primary endpoint is not
selected according to the resulting $p$-values.

\begin{table}[!htbp]
 \centering\small
 \setlength{\tabcolsep}{4pt}
 \caption{Separate-fit sandwich score results and original angular-comparator
 $p$-values. The first two rows concern the primary AX3 manual
 endpoint. Other endpoint comparisons are exploratory; intervals
 and $p$-values are not adjusted for multiple comparisons.}
 \label{tab:cataract_endpoint_sensitivity}
 \begin{tabular}{llrlrr}
\toprule
Endpoint / analysis & $n_S/n_V$ & $\widehat\Delta_\beta$ & 95\% score interval & $p_\beta$ & $p_W$ \\
\midrule
As supplied & 20/19 & -0.1805 & $[-0.7479,\ 0.3868]$ & 0.5328 & 0.4726 \\
Mean-entry exclusion & 19/18 & -0.1698 & $[-0.7544,\ 0.4148]$ & 0.5691 & 0.5038 \\
AX1\_manual & 20/19 & -0.1656 & $[-0.8126,\ 0.4815]$ & 0.6160 & 0.6208 \\
AX1\_SIA & 20/19 & -0.4817 & $[-1.0846,\ 0.1212]$ & 0.1174 & 0.0544 \\
AX3\_SIA & 20/19 & -0.5275 & $[-1.1538,\ 0.0988]$ & 0.0988 & 0.1014 \\
\bottomrule
\end{tabular}

\end{table}

Figure~\ref{fig:cataract_busemann} shows the fitted densities and
intervals for the two primary-endpoint analyses.

\begin{figure}[!htbp]
 \centering
 \includegraphics[width=\linewidth]{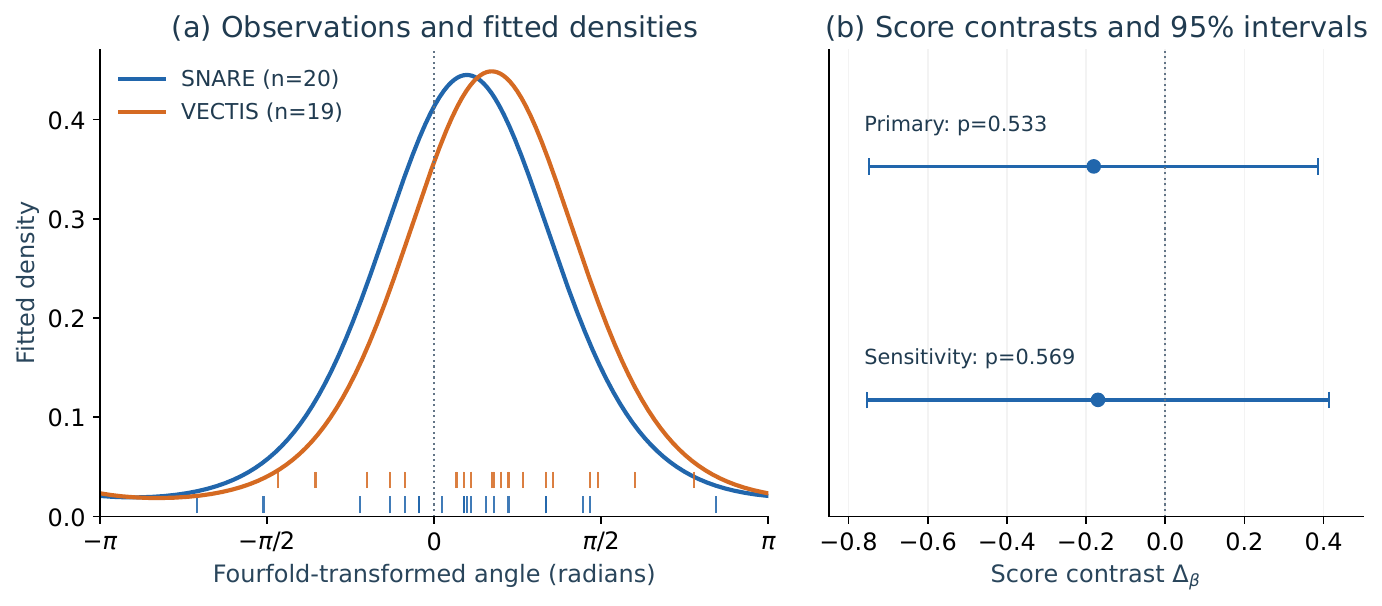}
 \caption{Primary cataract-axis analysis. Left: fitted von Mises
 densities and separate observation rugs on a shared angular scale;
 the vertical line marks the preferred transformed direction.
 Right: estimated score contrasts and nominal 95\% sandwich intervals for
 the primary analysis and its defined exclusion sensitivity.
 The vertical line marks the score-equality null.}
 \label{fig:cataract_busemann}
\end{figure}

Figure~\ref{fig:cataract_model_qq} compares ordered centered angles
$\Arg\exp\{i(\Theta_{gj}-\widehat\mu_g)\}$ with the corresponding
quantiles of a centered von Mises distribution at
$\widehat\kappa_g$. Parameters are estimated from the same
observations. These plots are descriptive checks of fitted model
shape; no calibrated goodness-of-fit $p$-value or reference envelope
is assigned. They do not establish the finite-sample level of the
score test at the observed sample sizes.

\begin{figure}[!htbp]
 \centering
 \includegraphics[width=.80\linewidth]{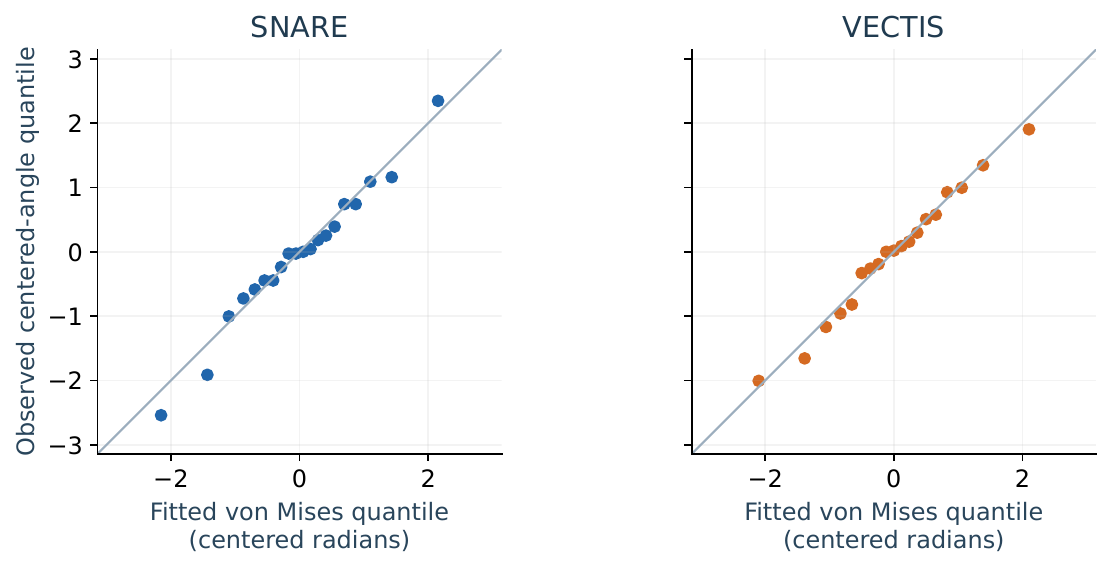}
 \caption{Descriptive fitted-model QQ plots for the primary
 observations. The reference line is the identity. Quantiles use
 the separate fitted concentrations; the plotted residual angles
 are centered at the fitted mean direction.}
 \label{fig:cataract_model_qq}
\end{figure}

% \subsection{Reproducibility}
% The reproducibility archive contains the extracts, audit, code, results and plotting data.
% Independent R and Python calculations reproduce the separate and pooled
% sandwiches and the original angular comparator, without resampling,
% imputation or critical-value adjustment.

\section{Simulation design and diagnostic results}
\label{app:simulation}
\subsection{Sampling and exact configurations}
\label{app:simulation_protocol}
The sample sizes and angle grids are defined in
Section~\ref{sec:simulation}. In the baseline study, each of 24 sampling units indexed by
$(N,\kappa_1,\kappa_2)$ used a separate L'Ecuyer--CMRG stream and
10,000 independent sample pairs generated by the
algorithm of \citet{best1979simulation}. Centered samples were reused across
mean pairs and procedures within a unit.
\begin{enumerate}
 \item \textbf{Unequal-concentration score criterion.}
 \label{app:simulation_score_design}
 For the null, set $c=\beta(\mu_0+\gamma_j,\kappa_L;\mu_0)$ with
 $\kappa_L=\min(\kappa_1,\kappa_2)$ and use both inverse-score
 partners at $\kappa_H=\max(\kappa_1,\kappa_2)$ in
 \eqref{eq:inverse_score}. Feasibility follows from
 $|c|\leq\rho(\kappa_L)<\rho(\kappa_H)$.
 For Figure~\ref{fig:sim_unequal}, set $\phi_1=0$ and
 $(\kappa_1,\kappa_2)=(1.5,12)$. The positive partner satisfies
 $\beta_2=\beta_1+d_\beta$ for $d_\beta=-1,-0.95,\ldots,1$.
 The contrast in \eqref{eq:hypotheses} is $-d_\beta$. The figure
 uses four independent streams, one per $N$, from seed 20261112,
 with 10,000 sample pairs reused across shifts within each size.
 \item \textbf{Common-concentration special case.}
 \label{app:simulation_common}
 Each common concentration uses the nulls $(\gamma_j,\gamma_j)$
 and $(\gamma_j,-\gamma_j)$, including alignment and opposition.
 The alignment power path is $(0,\gamma_j)$. Reflected angles
 are taken modulo $2\pi$; joint reflection gives the corresponding
 negative-anchor configurations.
 \item \textbf{Unequal-concentration benchmark controls.}
 For the angular null, use $(\gamma_j,\pm\gamma_j)$ at each
 ordered pair. For the cosine-mean null, set the less concentrated
 population at $\gamma_j$ and the other at
 $\pm\arccos\{A(\kappa_L)\cos\gamma_j/A(\kappa_H)\}$.
 The original-style benchmark demonstration uses $(0,\gamma_j)$
 for each ordered pair. The regular local alternatives are given
 in Appendix~\ref{app:simulation_local}.
\end{enumerate}

\label{app:simulation_reproducibility}
For a rejection rate $\widehat p$, the Monte Carlo standard error is
$\sqrt{\widehat p(1-\widehat p)/10000}$, about $0.00218$ near $0.05$.
Differences between procedures use paired rejection indicators. The
supplied scenario manifest and result CSVs contain the full grid and
pointwise uncertainties. Independent raw-data checks reproduced
1,512 fixture statistics to within $4\times10^{-10}$; there were no
invalid production evaluations.

\subsection{Comparator targets and results}
\label{app:simulation_targets}
Biswas $Z$ uses \eqref{eq:biswas_Z} and a normal reference; $W$
uses \eqref{eq:data_W_components}--\eqref{eq:data_W} and its
$\chi^2_1$ reference. The rank benchmark counts pairs
$d_{1i}<d_{2j}$, where $d_{gj}=|\Arg e^{i(\Theta_{gj}-\mu_0)}|$,
using the continuous Mann--Whitney null mean $n_1n_2/2$ and variance
$n_1n_2(N+1)/12$. Its reference assumes equality of folded-distance
distributions. The separate model-variance procedure (BS-U) and
restricted common-concentration diagnostics (Appendix~\ref{app:pooled_endpoints})
are retained in the output files.

Under unequal concentrations at $N=500$, $W$ rejects in
$4.08\%$--$5.41\%$ of replications across interior angular nulls, and
in $0.31\%$--$0.72\%$ at the joint endpoints, where its usual
reference is nonregular. On its own cosine-mean null, $Z$ rejects in
$4.38\%$--$5.42\%$. For the balanced design, its pooled cosine
variance consistently estimates the variance of the difference under
cosine-mean equality, even if the within-group variances differ.
These controls concern the comparators' respective targets.

As a separate comparator check, at $(\kappa_1,\kappa_2)=(1.5,3)$,
$\phi_1=0$, $d_\beta=0$ and $N=500$, BS-US rejects in $5.10\%$ of replications and $W$ in
$99.89\%$: here the score null holds and the angular null is false.
Figure~\ref{fig:sim_unequal_benchmark} illustrates the different targets
on the alignment benchmark. Reversing only the concentrations changes
that path and can remove its score-null point. Exchanging both groups
changes the sign of $\widehat\Delta_\beta$ and preserves its two-sided
$p$-value. Neither demonstration orders power for a shared hypothesis.
The reproducibility archive records the designs and results.

\begin{figure}[!htbp]
 \centering
 \includegraphics[width=.94\linewidth]{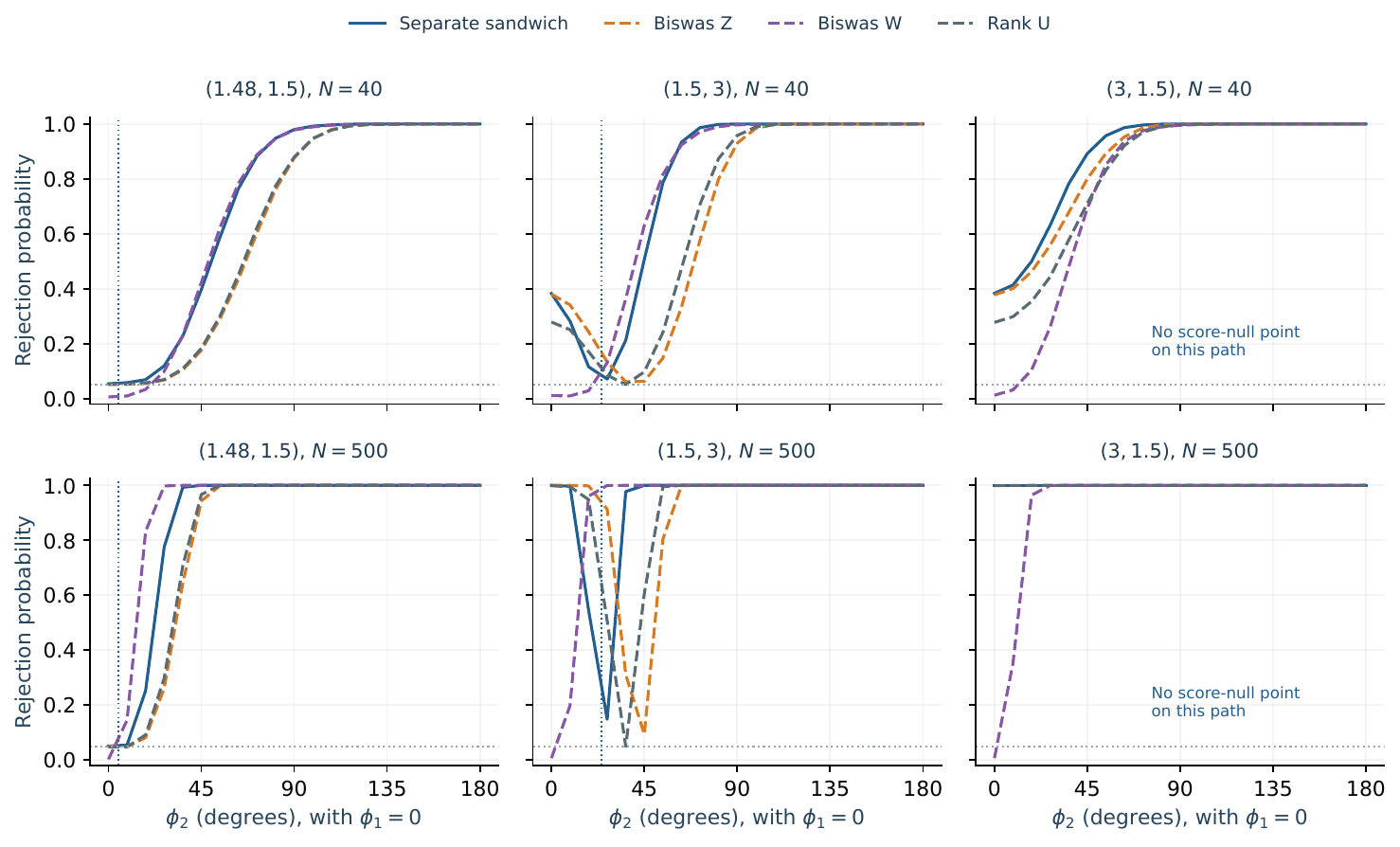}
 \caption{Unequal-concentration rejection curves on the Biswas
 benchmark path, at $N=40$ and $500$. The dotted vertical line marks
 score equality where feasible; there is no score-null point for
 $(\kappa_1,\kappa_2)=(3,1.5)$. The procedures test distinct criteria.}
 \label{fig:sim_unequal_benchmark}
\end{figure}
\FloatBarrier

\subsection{Coverage and local power}
\label{app:simulation_local}
Across the core score-null grid, the coverage of nominal 95\%
BS-US intervals was $91.40\%$--$95.05\%$ at $N=40$ and
$94.32\%$--$95.66\%$ at $N=500$. The corresponding ranges of the
mean estimated standard error divided by the empirical standard
deviation were $0.887$--$1.012$ and $0.982$--$1.019$.
Figure~\ref{fig:sim_qq} compares the centered statistic in
\eqref{eq:normal_limit} with the standard normal at representative
score nulls.

\begin{figure}[!htbp]
 \centering
 \includegraphics[width=.92\linewidth]{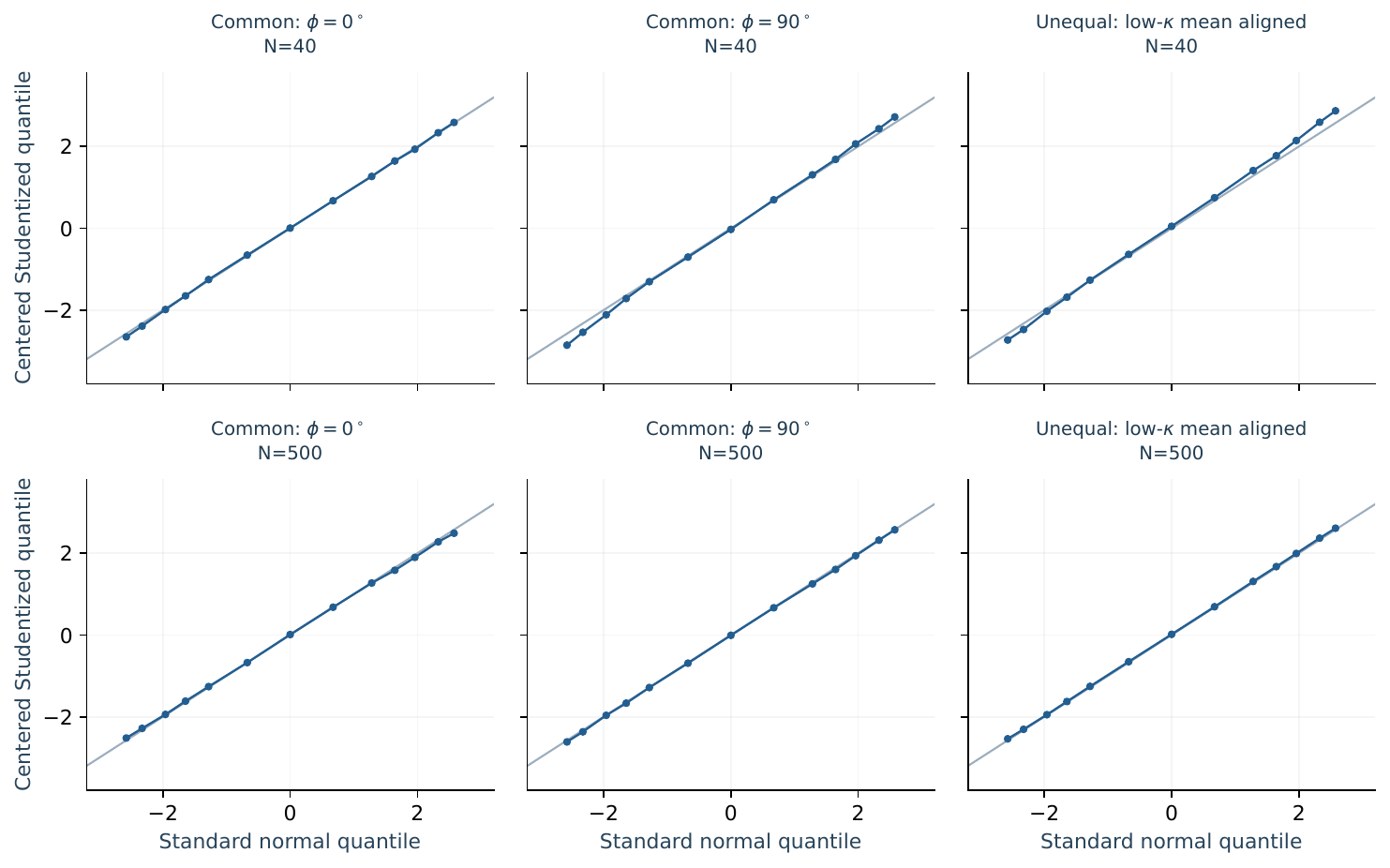}
 \caption{BS-US quantile comparisons at representative score nulls.
 Rows fix $N=40$ and $500$. The line is the standard normal reference.
 For the unequal case, $(\kappa_1,\kappa_2)=(1.5,3)$ and group 2
 takes the positive inverse-score partner.}
 \label{fig:sim_qq}
\end{figure}

For regular local power, we shifted the second mean from a fixed
score-null pair by $h/\sqrt N$, $h\in\{-4,-2,0,2,4\}$.
The base configurations were (L1) common $\kappa=1.5$ and
$\phi_1=\phi_2=\pi/4$; (L2) $(1.5,3)$ with $\phi_1=\pi/4$
and the positive partner; and (L3) $(3,1.5)$ with $\phi_2=\pi/4$
and the negative partner. Figure~\ref{fig:sim_local_power}
compares rejection rates with \eqref{eq:local_power}. At $N=500$,
the maximum absolute discrepancy across the four nonzero shifts
was between $0.0080$ and $0.0150$ across the three cases.
These differences include Monte Carlo error and finite-$N$ deviations.

\begin{figure}[!htbp]
 \centering
 \includegraphics[width=.96\linewidth]{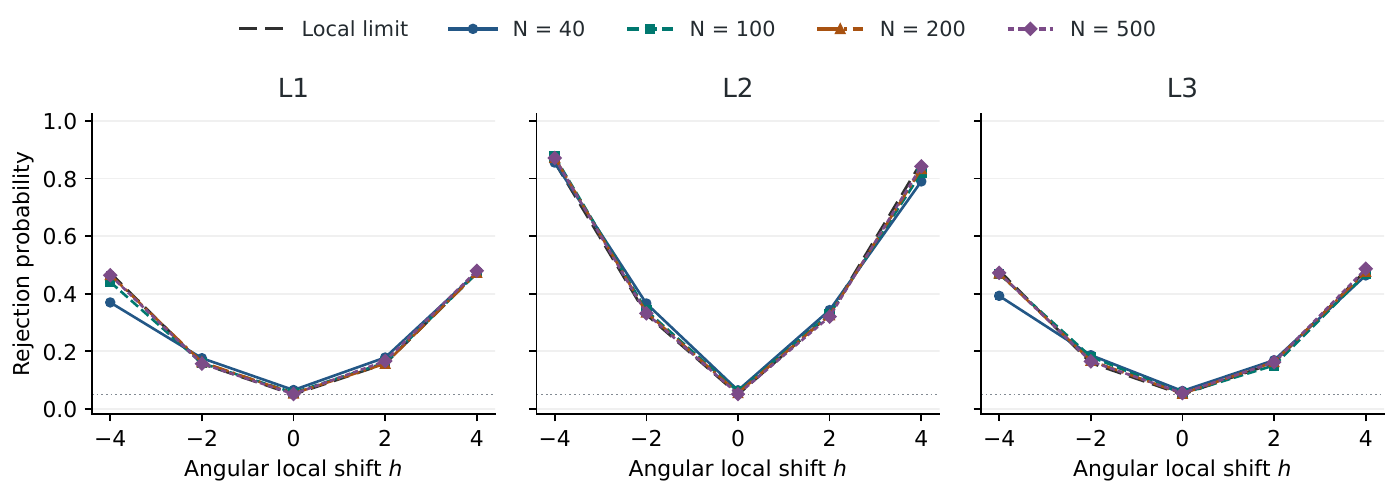}
 \caption{Regular local power for BS-US. The dashed curve is
 the limiting prediction at each fixed null; points show finite-sample
 rejection rates with pointwise 95\% Monte Carlo normal intervals.}
 \label{fig:sim_local_power}
\end{figure}
\FloatBarrier

\FloatBarrier
\end{document}